%% file: main-itcs.tex
\documentclass[a4paper,USenglish,cleveref, autoref, thm-restate]{lipics-v2021}
\usepackage{graphicx} 
\usepackage{caption} 
\usepackage{algorithm}
\usepackage{algorithmic}

\usepackage{amsmath,amsthm,amssymb}
\usepackage{mathtools}
\usepackage{dsfont}
\usepackage{xcolor}

\newtheorem{fact}{Fact}
\newtheorem{problem}{Problem}
\usepackage{thmtools}
\usepackage{thm-restate}

\usepackage{hyperref}

\usepackage{cleveref}

\input{preamble.tex}

\title{Clustering with Label Consistency}

\author{Diptarka Chakraborty}{National University of Singapore}{diptarka@comp.nus.edu.sg}{}{}
\author{Hendrik Fichtenberger}{Google Research, Zürich, Switzerland}{fichtenberger@google.com}{}{}
\author{Bernhard Haeupler}{INSAIT, Sofia University "St. Kliment Ohridski", Bulgaria \and
ETH Zürich, Switzerland}{bernhard.haeupler@inf.ethz.ch}{}{}
\author{Silvio Lattanzi}{Google Research, Barecelona, Spain}{silviol@google.com}{}{}
\author{Ashkan Norouzi-Fard}{Google Research, Zürich, Switzerland}{ashkannorouzi@google.com}{}{}
\author{Ola Svensson}{EPFL, Lausanne, Switzerland}{ola.svensson@epfl.ch}{}{}

\authorrunning{D.~Chakraborty, H.~Fichtenberger, B.~Haeupler, S.~Lattanzi, A.~Norouzi-Fard, and O.~Svensson} 

\Copyright{Diptarka Chakraborty, Hendrik Fichtenberger, Bernhard Haeupler, Silvio Lattanzi, Ashkan Norouzi-Fard, and Ola Svensson} 



\category{} 

\relatedversion{} 

\nolinenumbers 

\begin{document}

\maketitle
\begin{abstract}
    Designing efficient, effective, and consistent metric clustering algorithms is a significant challenge attracting growing attention. Traditional approaches focus on the stability of cluster centers; unfortunately, this neglects the real-world need for stable point labels, i.e., stable assignments of points to named sets (clusters). In this paper, we address this gap by initiating the study of label-consistent metric clustering. We first introduce a new notion of consistency, measuring the label distance between two consecutive solutions. Then, armed with this new definition, we design new consistent approximation algorithms for the classical $k$-center and $k$-median problems. 
\end{abstract}

\input{intro}

\input{prelim.tex}

\input{kcenter}

\input{kmedian}



\section*{Conclusion and Future Work}
We introduce the notion of label consistency for clustering and provide efficient approximation algorithms for the $k$-center and $k$-median problems. Natural future work could extend our approach to the $k$-means problem or improve the approximation ratio of our solution. In addition, in the current formulation of the label consistency problem, the label of a point is considered to be the center identifier of the cluster to which the point belongs. However, it could also be interesting to consider labels from a set of ids (e.g., $\{1,2, \cdots, k\}$). In this case, all points in the same cluster have the same label. The objective would then be to minimize the total number of label changes while keeping the cost low. We hope that many of the techniques (and results obtained) in this paper can be extended to this alternative variant.
\bibliography{refs}

\clearpage
\appendix
\input{app-k-cetner}
\input{app-fastdp}

\input{app-kmedian}
\input{app-integrality}
\input{app-lp-cost}


\end{document}

%% file: preamble.tex
\usepackage{microtype}
\usepackage{booktabs} 

\usepackage{hyperref}

\usepackage{amsmath}
\usepackage{amssymb}
\usepackage{mathtools}
\usepackage{amsthm}
\usepackage{optidef}
\usepackage{stmaryrd}
\usepackage{bm} 

\usepackage[textsize=tiny]{todonotes}

\usepackage{amsmath}
\usepackage{amssymb}
\usepackage{amsthm}
\usepackage{thmtools,thm-restate}
\usepackage{forloop}
\usepackage{bm}
\usepackage{bbm}
\usepackage{algorithm}
\usepackage{algorithmic}
\usepackage{mdframed}
\usepackage{bbm}

\usepackage{xspace}

\usepackage[l2tabu, orthodox]{nag}

\usepackage{hyperref}

\usepackage{nth}
\usepackage{color}

\usepackage{array}

\usepackage{paralist}
\usepackage{times}

\usepackage{import}

\usepackage{flafter}
\usepackage{pdfpages}

\usepackage{tikz}
\usetikzlibrary{calc}
\usetikzlibrary{matrix}
\usetikzlibrary{shapes}
\usetikzlibrary{decorations.pathmorphing}
\usetikzlibrary{decorations.pathreplacing}

\colorlet{darkgreen}{green!50!black}

\newcommand{\defcal}[1]{\expandafter\newcommand\csname c#1\endcsname{{\mathcal{#1}}}}
\newcounter{ct}
\forLoop{1}{26}{ct}{
    \edef\letter{\Alph{ct}}
    \expandafter\defcal\letter
}

\DeclareMathOperator{\poly}{poly}

\newcommand{\E}{{\mathbb{E}}}

\newcommand{\sgraph}[1]{\ensuremath{\mathcal{SG}
\ifthenelse{\equal{#1}{}}{}{(#1)}
}}
\newcommand{\cgraph}[1]{\ensuremath{\mathcal{CG}
\ifthenelse{\equal{#1}{}}{}{(#1)}
}}
\newcommand{\cpath}[2]{\ensuremath{P_{#1}
\ifthenelse{\equal{#2}{}}{}{(#2)}
}}

\newcommand{\OPT}{{\ensuremath{\mathrm{OPT}}}\xspace}

\newcommand{\cost}{{\ensuremath{\mathrm{cost}}}\xspace}
\newcommand{\con}{{\ensuremath{\mathrm{consistency}}}\xspace}
\newcommand{\clcst}{{\ensuremath{\mathrm{clcost}}}\xspace}

\newcommand{\id}{{\ensuremath{\mathrm{id}}}\xspace}
\newcommand{\C}{\mathcal{C}}
\newcommand{\stcon}{{\ensuremath{\mathrm{str{-}consistency}}}\xspace}

\newcommand{\swcost}{{\ensuremath{\mathrm{swcost}}}\xspace}
\newcommand{\fol}{{\ensuremath{\mathrm{Fol}}}\xspace}

\newcommand{\dav}{{\ensuremath{d_{\mathrm{av}}}}\xspace}
\newcommand{\ball}{{\ensuremath{\mathrm{Ball}}}\xspace}
\newcommand{\vol}{{\ensuremath{\mathrm{Vol}}}\xspace}

\newcommand{\polylog}{\poly\log}
\newcommand{\eps}{\varepsilon}

%% file: intro.tex
\section{Introduction}

Clustering is a fundamental problem in unsupervised learning with wide-ranging real-world applications. Its core principle involves grouping similar objects together while separating dissimilar ones. Over the years, numerous clustering formulations have been explored. In this paper, we focus on two classic metric clustering problems, specifically the $k$-center and $k$-median problems. Formally, given a set of points $P$ from a metric space $\mathcal{X}$, let $d(i,j)$ denote the distance between $i, j \in P$. For a set of centers $C \subseteq P$, and an assignment of points from $P$ to $C$, denoted by $\mu$, the $k$-center cost is defined as $\cost_{cen}(P, (C, \mu)) = \max_{j \in P} d(\mu(j),j)$ and the cost is $\cost_{med}(P, (C, \mu)) = \sum_{j \in P} d(\mu(j),j)$ for the $k$-median problem.
In both problems, one is interested in finding $k$ cluster centers and an assignment that minimize the respective costs.

Thanks to their simple and elegant formulation, the $k$-center and $k$-median problems have been extensively studied throughout the years. Efficient approximation algorithms are known for them \cite{charikar1999constant,jain2001approximation,ahmadian2019better,byrka2017improved} and also streaming \cite{charikar2003better}, sliding window \cite{braverman2015clustering, braverman2016clustering, borassi2020sliding, epasto2022improved, woodruff2023near} and dynamic algorithms have been proposed \cite{chan2018fully, HenzingerK20, bateni2023}. Despite all this work, until recently, not much was known about the stability and consistency provided by different clustering solutions. Notably, even small variations in input data could lead to vastly different clustering outputs across previously cited methods. This lack of consistency poses a serious concern for real-world applications where inputs evolve over time, and downstream processes depend on reliable, high-quality clustering.

In response to this issue, the past few years have seen the development of several new methods for consistent clustering \cite{lattanzi2017consistent, cohen2022online, fichtenberger2021consistent, cohen2019fully, bhattacharya2022efficient, lkacki2024fully, guo2021consistent, bhattacharya2024fully, forster2025dynamic, bhattacharya2025fully}. Interestingly, previous work has focused on algorithms with good center consistency (where cluster centers change smoothly) but offers no significant guarantees on label consistency (the stability of individual point assignments). However, many practical scenarios demand label consistency. For example, consider trust and safety applications where clustering is used to detect coordinated attacks. In such a scenario, stable labels, i.e., stable assignments of points to the same \emph{named} set, are crucial for reliable and consistent detection. Similarly, in machine learning pipelines, clustering often provides labels for training data; in this setting, label stability is essential for consistent model predictions. Finally, consider ontology or taxonomy tasks; also, in this scenario, one is interested in having a consistent label for each object. In all of these scenarios, it is common that points are treated based on the cluster they belong to. Therefore, a stable assignment to the same cluster guarantees a consistent treatment of points.

The motivation of this work is the aforementioned importance of \emph{label} consistency and the observation that current solutions do not provide any guarantees on label consistency. In fact, it is not hard to construct instances where algorithms that do not focus on label consistency change the label of several nodes needlessly. 
A simple example illustrating this point is the following $k$-median instance where points are in a single dimension and $k=2$. There are $10$ points at position $-2$, $1000$ points at $0$,  $10$ points at position $2$, and a single point at position $100$.  It is not hard to see that the optimal solution for this instance is to select a center in position $0$ and a center in position $100$. Now, suppose we add another set of $1000$ points in position $3$. To have a good solution with good label consistency, one should create a new cluster containing the new points and the initial center at $1000$ (with a center at position $3$) and leave untouched the old cluster around $0$. Unfortunately, this is not the case; algorithms focusing only on center consistency will reassign several points in the original cluster to the new one (specifically the points at $2$), creating instability in the solution. 
This is a simple example that illustrates that it may be suboptimal to assign points to their closest centers when optimizing for label consistency.
Another potential issue with center consistency is that the consistency of all centers are of ``equal'' importance and not weighted by the size of the respective clusters, which is a more relevant (weighted) measure in many applications.


\paragraph*{Label Consistent Clustering} To address these issues, we introduce the notion of label consistency in the incremental setting. In particular, given two sets of points $P_1$ and $P_2$, where $P_2$ is obtained from $P_1$ by adding points, and a clustering of $P_1$, we are interested in computing a high-quality clustering of $P_2$ with good label consistency.

More formally, given two set of points $P_1$ and $P_2$, where $P_1 \subseteq P_2$, and a solution $\C_1=(C_1,\mu_1)$ for $P_1$, the goal is to compute a high-quality solution $\C_2=(C_2,\mu_2)$ for $P_2$ (with respect to the cost function) consistent to $\C_1$. We define \emph{switching cost} or \emph{inconsistency} between $\C_1$ and $\C_2$, denoted by $\swcost(\C_1,\C_2)$, as the number of points in $P_1$ that are assigned to a different center by $\mu_1$ and $\mu_2$, i.e.,
$
\swcost(\C_1,\C_2) := \left| \left\{ i \in P_1 \mid \mu_1(i) \neq \mu_2(i) \right\} \right|.
$
The \emph{consistency} between $\C_1$ and $\C_2$ is defined as $\con(\C_1,\C_2):=|P_1| - \swcost(\C_1,\C_2)$.

We are now ready to formally define our new label consistent clustering problem.
\begin{definition}[Label Consistent Clustering Problem]
    \label{def:conscluster}
    Given two instances $P_1$ and $P_2$ of the $k$-center (median) problem, where $P_1 \subseteq P_2$, a solution $\C_1 = (C_1,\mu_1)$ of the instance $P_1$, and a budget $S$ on the switching cost, the objective is to find a solution $\C_2 = (C_2,\mu_2)$ for $P_2$ that minimizes the objective function $\cost(P_2,\C_2)$
    while satisfying the constraint $\swcost(\C_1,\C_2) \le S$, i.e.,
    \[
    \arg \min_{\C:\swcost(\C_1,\C) \le S} \cost(P_2,\C).
    \]
    Further, we use $\C_2^*$ to denote an optimum solution, and we let $\OPT = \cost(P_2, \C_2^*)$ denote the value of an optimal solution when the input is clear from the context.
\end{definition}


We remark that the label consistent $k$-center and $k$-median problems generalize the basic variants (by setting $P_1 = \emptyset$), and so they are hard to approximate better than a factor of $2$~\cite{HochbaumShmoys86} and $1+2/e$~\cite{GUHA1999228}, respectively. This motivates us to find an efficient approximation algorithm for the problem. \Cref{fig:example} depicts the difference between cluster consistency and label consistency. 

\begin{figure}[H]
    \centering
    \includegraphics[width=10cm]{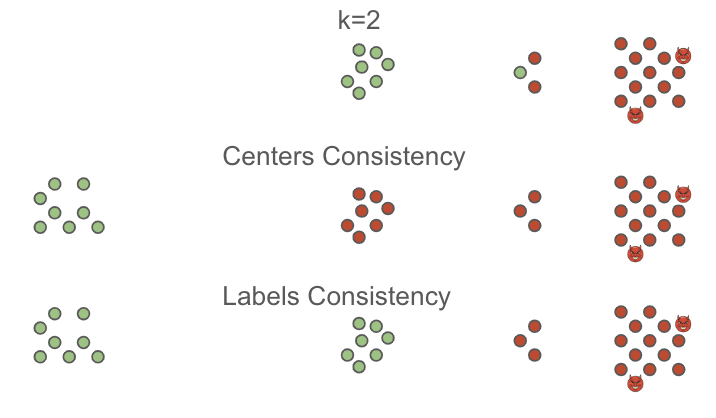}
    \caption{In the above example, we show different effects of the label or center consistency constraints. In this setting $k=2$ and the initial set of points is the one at the top. A good clustering for the $k$-center and $k$-median objective is the one represented by green and red colors at the top. Note that in the red cluster we have some bad actors, and so an algorithm that uses this clustering is likely to label all red points as bad. Suppose now that a new clustering is added and we require an algorithm to be $1$-center consistent or $1$-label consistent. The $1$-center consistent algorithm will completely change the labels of the initial points and mark all the input points as malicious. The 1-label consistent algorithm will instead just change the label of one point (the point for which it is more certain) and maintain the initial labels for other points. Thus, label consistency is providing a more stable and reliable solution.}
    \label{fig:example}
\end{figure}

\paragraph*{Our contribution and techniques}
The main contribution of our paper is to design approximation algorithms for the Label Consistent Clustering Problems. We begin with the classical $k$-center problem and provide a constant-factor approximation to its label-consistent variant.

\begin{restatable}{theorem}{kcenter}
    \label{thm:k-center}
    There is an $O(n^2 + kn \log n)$-time $6$-approximation algorithm for the label consistent $k$-center clustering problem, where $n$ is the number of points in the input instance.
\end{restatable}

The main idea behind our $k$-center clustering is to develop a two-phase algorithm. First, we expand clusters around centers in $C_1$. Second, we carefully designate new centers from the point set $P_2 \setminus P_1$ and selectively close old centers, reassigning their points to other existing centers. The key insight in the proof is that it is possible to strike a good balance between the number of new cluster centers, reassignments, and the quality of the clustering solution. We provide a more detailed overview in Section~\ref{sec:k-center}.

Next, we consider the label consistent variant of the $k$-median clustering, another fundamental $k$-clustering problem.

\begin{restatable}{theorem}{thmdpmain} \label{thm:dpmain}
There exists a $O((nk+k^3) \polylog n)$ time algorithm that  given an instance $P_1, \C_1 = (C_1, \mu_1), P_2, S$  of the label consistent $k$-median problem,  returns a  solution  $\C_2 = (C_2, \mu_2)$ whose switching cost is at most $S$ and  $\cost(P_2, \C_2) \leq O(\log k) \cdot (\OPT + \cost(P_1, \C_1))$ with probability at least $1-1/n^{10}$, where $n$ is the number of points in the input instance. 
\end{restatable}

We obtain the above result by first simplifying the metric structure of the point set using probabilistic tree embedding~\cite{fakcharoenphol2003tight} and then designing a dynamic program that utilizes this tree structure. Our main technical contribution in this part is to reduce the running time of the algorithm. In fact, a naive implementation of our approach would require time linear in the aspect ratio and cubic in the number of nodes, but thanks to our optimization, we manage to obtain an almost linear time algorithm. We refer to Section~\ref{sec:dp} for details. It is evident from our technique that we can easily avoid the $O(\log k)\cost(P_1, \C_1)$ additive factor by paying $O(\log n)$ multiplicative factor.

Next, we demonstrate that we can improve the $O(\log k)$ factor to a constant by allowing resource augmentation.

\begin{restatable}{theorem}{thmlpmain} There is a polynomial-time algorithm that  given an instance $P_1, \C_1 = (C_1, \mu_1), P_2, S$  of the label consistent $k$-median problem,  returns $(k+1)$ centers with  assignment  $\C_2 = (C_2, \mu_2)$ whose switching cost is at most $S$ and  $\cost(P_2, \C_2) \leq 10 \cdot \OPT + \cost(P_1, \C_1)$. 
Furthermore, for every fixed constant $\varepsilon > 0$, there is a polynomial-time algorithm that returns $k$ centers with assignment $\C_2 = (C_2, \mu_2)$ whose switching cost is at most $(1+\varepsilon)S$ and  $\cost(P_2, \C_2) \leq 10 \cdot \OPT + \cost(P_1, \C_1)$.
\label{thm:k-median}
\end{restatable}

The two $10$ multiplicative and $\cost(P_1,\C_1)$ additive approximation algorithms for the $k$-median clustering with resource-augmentation in either the number of clusters or allowable switching cost are obtained by considering an LP relaxation of the problem and then we develop a rounding that builds upon the filtering and bundling technique of~\cite{CharikarL12} and carefully incorporates additional steps to handle the consistency constraints. We provide a more detailed overview in Section~\ref{sec:lprelaxation}.

We complement these algorithmic results by showing that the standard linear programming relaxation has an unbounded integrality gap when one opens at most $k$ centers and must satisfy the switching cost exactly, establishing the necessity of resource-augmentation for the current LP relaxation approach (see~\cref{sec:lpintegralitygap}).

\noindent \textbf{Comparison with a concurrent work.}
Concurrently with our work, Gadekar, Gionis, and Marette~\cite{gadekar2025} study the label consistent $k$-center clustering problem and obtain a $2$-approximation in time $2^{k}\,\mathrm{poly}(n)$, as well as a $3$-approximation in polynomial time. The algorithms for label consistent $k$-center clustering in both our work and~\cite{gadekar2025} employ a similar greedy strategy. In addition, our work considers the label consistent $k$-median clustering problem and provides approximation results for it.

%% file: prelim.tex
\paragraph*{Additional notations} 

Given two instances $P_1$ and $P_2$ of the $k$-center (median) problem, where $P_1 \subseteq P_2$, a solution $\C_1 = (C_1,\mu_1)$ of the instance $P_1$, and a budget $S$ on the switching cost, we say a solution $\C_2 = (C_2,\mu_2)$ of $P_2$ is an $(\alpha,\beta,\gamma)$-approximate solution, for an $\alpha,\beta,\gamma \ge 1$, iff  $\swcost(\C_1,\C_2) \le \alpha \cdot S$, and $\cost(P_2,\C_2) \le \beta \cdot  \cost(P_2,\C_2^*)+\gamma $
where $\C^*_2=(C^*_2,\mu^*_2)$ denotes an (arbitrary) optimum solution for $P_2$ that satisfies the budget of at most $S$ on the switching cost. For brevity, when $\alpha = 1$, we simply call it a $(\beta,\gamma)$-approximate solution and, when in addition, $\gamma = 0$, we will simply call it a $\beta$-approximation.

Another notation that will be useful in describing our algorithms (and analysis) is the \emph{ball} of a certain radius around a point in the metric space. Let $X$ be the underlying metric space. For any point $x \in X$ and any $r \ge 0$, let $\ball(x,r)$ denote the \emph{ball} of radius $r$ around $x$. Formally,
$
\ball(x,r) := \left\{y \in X \mid d(x,y) \le r\right\}.
$

The \emph{aspect ratio}, denoted by $\Delta$, is defined as the ratio between the largest distance and the smallest non-zero distance between any two input points. As common in the literature, we assume that the aspect ratio of a set of $n$ input points is at most some polynomial in $n$. This assumption can be removed by adding a factor $O(\log \Delta)$ to the running time.

%% file: kcenter.tex
\section{Algorithm for Label Consistent $k$-Center}
\label{sec:k-center}
In this section, we provide an efficient approximation algorithm for the label consistent $k$-center. The following theorem, whose proof is postponed to Appendix~\ref{sec:app-k-center}, is the main result of this section:

\kcenter*

Our algorithm works in two phases. In the first phase, we grow the cluster around each "old" center ($c\in C_1$). Then, we create new centers from the point set $P_2 \setminus P_1$ that are not covered by the grown clusters around old centers. In the next phase, we close some of the old centers by reassigning the corresponding cluster points to other old centers. Let us now briefly elaborate on these two phases.

\noindent \textbf{Phase 1: Growing phase and opening new centers. }For ease of presentation, we assume that we know the optimum value of the objective function. (Later in Remark~\ref{rem:assump}, we discuss how to remove this assumption.) Let $R := \cost(P_2, \C^*_2)$, where $\C^*_2=(C^*_2,\mu^*_2)$ denote an (arbitrary) optimum solution for $P_2$ that satisfies the budget of at most $S$ on the switching cost. 

Now, we grow the clusters around each center $c \in C_1$ and cover some of the new points ($\in P_2 \setminus P_1$). For that purpose, for each center $c \in C_1$, consider the set $P_2 \cap \ball(c,2R)$ and call these points of $P_2$ \emph{covered}. Let $U$ denote the set of \emph{uncovered} points of $P_2$. 

Next, we open a set of new centers $C_2$ greedily from $U$ to cover all the points in $U$. For that purpose, we open a center at an arbitrarily chosen point $u \in U$, remove all the points in $\ball(u,2R)$ from $U$, and continue until $U$ becomes empty. Let $k'$ be the number of centers opened at the end of the above greedy procedure (i.e., $k' = |C_2|$). Thus, we are allowed to keep at most $k-k'$ centers from the set of old centers $C_1$. In the next phase, we keep $k-k'$ centers from $C_1$ and close the remaining.

\noindent \textbf{Phase 2: Closing and reassigning phase. }For each $c \in C_1$, let us define its \emph{weight} $w_c$ to be the size of the corresponding cluster in $\C_1$, more specifically, the number of points of $P_1$ that are assigned to the center $c$ by the assignment function $\mu_1$ and are within distance $2R$ from $c$.

Next, we keep (at most) $k-k'$ old centers open and close the remaining, then reassign the points in the old clusters of the closed centers. We start with all the centers in $C_1$ as unmarked. We then sort the centers $c \in C_1$ in the non-increasing order of their weights $w_c$. Next, we process them one by one in that sorted order. If a center is unmarked so far, we keep that center "temporarily" open (by adding them in a new set $T$) and mark (and close) all the centers $c' \in C_1$ such that $c' \in \ball(c, 2R)$. Next, we process each temporarily opened center $c \in T$. Consider the set $\ball(c,R)$ and let $c'$ be a center in $C_1\cap \ball(c,R)$ that has the highest weight among all the centers in $C_1 \cap \ball(c,R)$ (note that $c'$ could be different from $c$ because it could be dominated by a third center in $T$). Then, we add $c'$ to $C_2$ and remove $c$ from $T$. In the end, if the number of open centers is less than $k$, we open the closed centers (from $C_1$) one by one in the above-mentioned sorted order until the number of open centers is exactly $k$.

We describe the algorithm in detail (along with pseudocode) in~\cref{sec:app-k-center} and present a visual description of it in Figure~\ref{fig:kcenter}. Let $C_2$ denote the set of centers opened by our algorithm. We define the assignment function $\mu_2$ for the center set $C_2$ as follows: For each point $p \in P_2$, if $\mu_1(p) \in C_2$ and $d(p,\mu_1(p)) \le 2R$, then set $\mu_2(p) = \mu_1(p)$; otherwise, set $\mu_2(p)$ to be a closest (breaking ties arbitrarily) center point in $C_2$.

\begin{figure}[h]
    \includegraphics[scale=1]{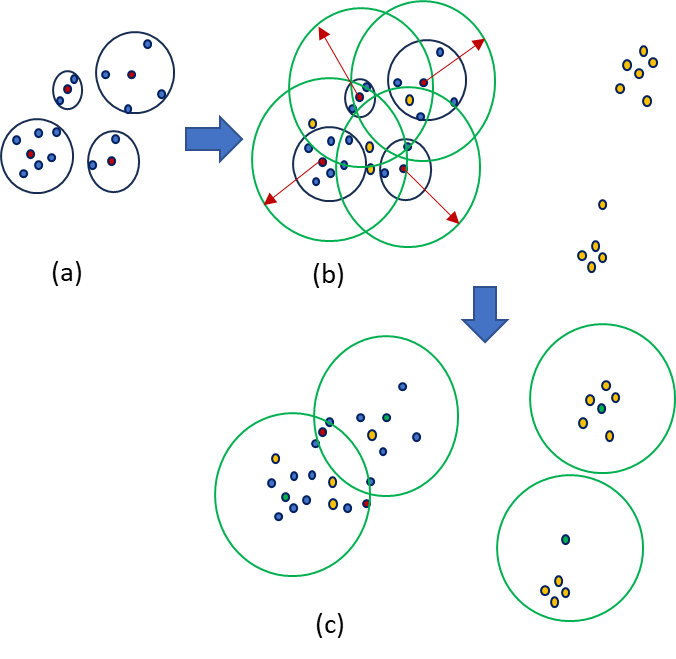}
    \caption{\label{fig:kcenter}Visualization of the k-center clustering algorithm: {\bf{(a)}} A solution of $4$-center (i.e., $k=4$) on point set $P_1$ (red points denote centers). {\bf{(b)}} Yellow points are added, i.e., $P_2 = P_1 \cup $ \{set of new yellow points\}. Balls (of color green) of radius $2R$ (where $R$ is the optimum $k$-center cost of $P_2$ with switching cost at most $S$) are drawn around old (red) centers, and the yellow points that are uncovered (by green-colored balls) constitute the set $U$. {\bf{(c)}} Two new centers are opened to cover points in $U$, and thus two old centers need to be closed. Finally, four green points are the four centers (among them two are old centers and two are new) for the final $4$-center solution for $P_2$. Note that the points that are covered by multiple green clusters can be assigned to their closest green center points.
    }
\end{figure}

Next, we sketch a proof of~\cref{thm:k-center}; for a full proof, please refer to~\cref{sec:app-k-center}.

\begin{proof}[Proof Sketch of~\cref{thm:k-center}] 
We compare the solution output by our algorithm with an (arbitrary) optimum solution $\C^*_2=(C^*_2,\mu^*_2)$ for $P_2$ such that $\swcost(\C_1,\C^*_2) \le S$. Let $O$ be the set of points in $P_2$ that could be covered by the old centers $C_1$, and $N$ be the set of remaining points. We first argue that in Phase 1, from the set $N$, our algorithm opens at most as many centers opened by an optimal solution $\C^*$ (\cref{lem:comp-new-center}). Using that, we argue that our algorithm opens at most $k$ centers. Next, we show that for each center $c^*$ in an optimal solution ($c^* \in C^*_2$), there is a center opened by our algorithm within a distance of $5R$, and thus by triangle inequality, for each point $p \in P_2$, there is a center in $C_2$ within a distance of $6R$. (See~\cref{lem:k-cen-feasible} for the detailed argument.) It only remains to show that $\swcost(\C_1, \C_2) \le S$. For that purpose, we argue that $\swcost(\C_1, \C_2)$ is not more than $\swcost(\C_1, \C^*_2)$ which is at most $S$. The crux of the argument lies in the fact that for each point in $T$ (i.e., each temporarily opened center), if we consider a ball of radius $R$ around it, an optimal solution $\C^*_2$ must open a distinct center within that ball. Now, since our algorithm opens a center of the highest weight from each such ball, the total weight (i.e., consistency) is at least that achieved by the optimal solution $\C^*_2$. (See~\cref{lem:k-cen-switch}.)
\end{proof}

%% file: kmedian.tex
\section{Approximation Algorithms for Label Consistent $k$-Median}

In this section, we present label consistent clustering algorithms for the $k$-median objective. Our first algorithm, presented in Section~\ref{sec:dp}, uses a probabilistic tree embedding to embed the metric in a tree and then uses a carefully design dynamic programming approach to solve the problem.

\thmdpmain*

Our remaining results for $k$-median are based on the linear programming (LP) relaxation for the label-consistent clustering problem and are presented in Section~\ref{sec:lprelaxation}. 

\thmlpmain*

We complement these algorithmic results by showing that the standard linear programming relaxation has an unbounded integrality gap when one opens at most $k$ centers and must satisfy the switching cost exactly. This shows that this LP cannot be used to give an exact algorithm and, therefore, shows the tightness of our results, which use a minimally larger switching cost or one extra center to achieve a good approximation. We present the following result in~\cref{sec:lpintegralitygap}.

\begin{lemma}
\label{lem:integrality}
    The standard linear programming relaxation for label consistent $k$-median clustering has an unbounded integrality gap.
\end{lemma}

\input{dp}

\input{lp}

%% file: dp.tex
\subsection{Approximation for $k$-median via Probabilistic Tree Embedding and Dynamic Programming}\label{sec:dp}

Our approximation algorithm for $k$-median employs various carefully chosen simplification, and rounding steps until the simplified problem can be solved approximately via a low-dimensional dynamic program that can be efficiently run on large data sets. 

A key simplification step in our algorithm is a probabilistic tree embedding that simplifies the structure of distances between points. 

\begin{definition}[Tree embedding]
Given a set of $t$ points, a tree embedding is a rooted binary tree with the following guarantees:\\
-- There are exactly $t$ leaves, each corresponding to one input point. \\
-- All internal nodes have exactly two children.\\
The depth of the tree is denoted by $\ell_{max}$.
\end{definition}

The internal nodes of a tree embedding naturally correspond to subsets of points, namely the points corresponding to the leaves in the sub-tree rooted at the internal node. We label each internal node with the diameter of this associated point set. Naturally, these tree-distances are monotone as one goes down in the tree. The tree-distance of two points in a tree embedding is defined as the diameter of the smallest cluster containing both or, equivalently, the distance label given to the lowest common ancestor of the two leaves to which they correspond. It is clear by definition that the tree-distance of two points is always at least as large as their distance. A good probabilistic tree embedding guarantees that the distances are not stretched too much in expectation. Many efficient algorithmic constructions of such probabilistic tree embeddings with stretch $O(\log t)$ are known~\cite{fakcharoenphol2003tight}. 

\begin{theorem}\cite{fakcharoenphol2003tight}
There exists an algorithm that, given any $t$ points and their $s=t^2$ metric distances with an aspect ratio of at most $\Delta$, samples in $O(s \log s \log \Delta)$ time, a tree embedding with depth $\ell_{max} = O(\log s \log \Delta)$ and expected stretch $O(\log t)$, i.e., with the guarantee that for any two points their tree-distance is not lower and in expectation at most an $O(\log t)$ factor larger than their distance. 
\end{theorem}

Our algorithm starts with reducing the number of input points, before applying the tree-embedding to ensure the linear running time dependency on the size of the input. 
A common approach for such an input size reduction while losing only a small factor in approximation guarantees is to use coresets. Unfortunately it is not clear how those type of solutions can be used in the context of label-consistency since they merely approximately preserve costs of solutions by sampling and reweighting points but subtly yet crucially do not guarantee assignments of non-sampled points to weighted sampled points. 

To avoid this issue, we use a moving-based strategy. We start by first moving all the points in the old solution ($C_1$) to their centers and let the weight of each center equal the number of points in its cluster. To reduce the size of the new points in $P_2 \setminus P_1$, we use any (non-consistent) constant approximation algorithm $\mathcal{A}$ that computes a clustering of size $O(k)$ and run it on $P_2 \setminus P_1$. One well-known such algorithm is k-median++~\cite{arthur2007k}, which is a constant approximation when used with $O(k)$ centers. Afterward, we move all the points to their assigned centers. For the points in $P_1$, we move them to the center that they are assigned to in $C_1$. We let the weight of a center be the number of points moved to the center in both cases. 

To summarize, we used a moving approach that creates an instance of $O(k)$ weighted new points and $k$ weighted old points while increasing the cost by at most a constant factor plus the cost of the old solution ($C_1$).

As a next step, we embed all these (weighted) points into a tree. This only increases distances and, therefore, connection costs but otherwise leaves switching costs of solutions unchanged. Since distances change in expectation by at most an $O(\log k)$ factor point-wise, the expected cost of the optimal solution changes at most by this $O(\log k)$ factor in expectation and remains feasible.

In the remainder of this section, we show how to efficiently compute a constant-approximate solution for the resulting simplified problem with $\Theta(k)$ points on a tree using dynamic programming and rounding, resulting in a randomized algorithm that produces a solution with expected $O(\log k)$ cost. Repeating this $O(\log n)$ times and taking the best solution guarantees a $O(\log k)$-approximation with high probability, e.g., with probability at least $1 - 1/n^{10}$.

We first show that the problem of $k$ nodes of total weight summing up to $n$ with distances given by a tree embedding can be solved exactly using dynamic programming, albeit with a prohibitively large in $O(n^2 k^3)$ running time, which still depends on $n$. We then show that clever rounding and rearrangements can reduce the running time to $O(k^3 \polylog n)$.

\begin{theorem}\label{thm:ExactTreeDP}
Given $t$ nodes with total weight $n$ whose distances are given by a tree embedding, the optimal $k$-median cost for switching cost $S\leq n$ can be computed deterministically and exactly in $O(t \cdot S^2 \cdot k^2) = O(n^2 k^3)$ time. 
\end{theorem}
\begin{proof}
Our exact dynamic program has three dimensions. The first one indicates the id of the node on the tree, the second one is the number of centers we allow to be opened in the subtree of the node, and the last one is the consistency of the subtree, i.e., the number of points assigned to their old center. Notice that since we moved all the old points to their center, if an old center is open, all the points in its cluster are assigned to it. The goal of the DP is to find the minimum cost solution that assigns all the points to a center in the subtree with the discussed requirements. If no such solution exists, we define the cost as infinite.
\begin{align*}
dp[\text{id}][k][c] = \text{Min cost solution in subtree of node 'id',} \\
\text{ with at most $k$ centers, and consistency at least $c$.}
\end{align*}

The initialization of the DP for the leaves is straightforward. For any $k \ge 1$, if $c = 0$, then we open the point on this leaf, and the cost of the solution is zero; otherwise, the cost is zero if the point on this leaf is an old point with weight equal or higher than $c$, and infinite otherwise. Recall that we have at most one point assigned to each leaf by the properties of the tree-embedding.

For the non-leaf nodes of the tree, one needs to distribute the number of open centers and consistency between the two children of the node.  
We next describe how to compute $dp[\text{id}][k][c]$ for a node ``id'' with its two children named ``l'' and ``r'' and values $k' \leq k$ and $c$. The first option is to open zero centers in ``l'' subtree and connect all points to one of the centers opened in ``r'' subtree. The cost in this case is
\begin{align*}
     dp[\text{r}][k][c] + \text{(num nodes in ``l'')} &\cdot \text{(diameter of the points of ``id'').} 
\end{align*}
Recall that the diameter of the points of node ``id'' is an upper bound on the maximum distance in its subtree computed by the tree embedding. An analog holds for the case that we do not open any center on ``r''. The only remaining options are to open $1 \leq k' \leq k-1$ centers on ``l'' with consistency $0 \leq c' \leq c$, and $k-k'$ centers with $c-c'$ consistency on ``r''. The costs for these options are

\[
dp[\text{l}][k'][c'] + dp[\text{r}][k-k'][c-c'].
\]

The value of $dp[\text{id}][k][c]$ is the lowest among all discussed options. 

The total number of values computed by the DP is at most $O(t \cdot k \cdot S)$ because the tree consists of $O(t)$ nodes; the second dimension counting the open centers can take on values between $1$ and $k$, and the last dimension specifying the consistency can only take on values between $|P_1| - S$ and $|P_1|$. The number of options considered by the DP above for every node is, furthermore, at most $(k+1) \cdot (S+1)$. Overall, at most $O(t k^2 S^2)$ additions and comparisons are performed by the DP. 
\end{proof}

The exact DP stated in Theorem~\ref{thm:ExactTreeDP}, with its quintic dependency on input parameters and quadratic dependency on $S$, which is often of size proportional $n$, i.e., the number of points in $P_2$ is much too computationally expensive.

To improve the running time, we rearrange the DP and round connection costs and utilize the fact that the tree obtained from the tree embedding has low depth. This produces a solution that remains exact in its consistency and has a fast $O(k^3 \polylog n)$ running time depending on $n$ only polylogarithmically. A detailed description of this optimized DP (along with a pseudocode of the whole algorithm) and its analysis is given in~\cref{sec:dpfast}.

\begin{theorem}\label{thm:RoundedTreeDP}
There exists a randomized algorithm that given $t = O(k)$ weighted points with weights and aspect ratios bounded by some polynomial in $n$ and a switching cost constraint $S$, computes an $O(\log k)$-approximate $k$-median clustering in $O(k^3 \polylog n)$ time with high probability, i.e., with probability at least $1 - 1/n^{10}$.
\end{theorem}

The main result from~\cref{thm:dpmain} follows now almost immediately, given the discussed steps in this section and is presented in \cref{app:dpmainproof}.


%% file: lp.tex
\subsection{Linear Programming Relaxation}\label{sec:lprelaxation}

Recall that $\C_1 = (C_1, \mu_1)$ denotes the set $C_1$ of centers  opened in the first stage, and the assignment $\mu_1$ of the points $P_1$ to these centers.
We let $w_i$ denote the number of points of $P_1$ that are assigned to $i\in C_1$.
Our goal is to open $k$ centers $C_2$  and find an assignment $\mu_2 : P_2 \rightarrow C_2$ that achieves the following:\\
-- The switching cost \(\swcost(\C_1,\C_2) \) is upper bounded by a given target $S$.\\
-- The connection cost \(\sum_{j \in P_2} d(\mu_2(j), j))\) is minimized.

To formulate this as a linear program, we have a variable $y_i$ for each potential center location $i\in P_2$. The intuition is that $y_i = 1$ if $i\in C_2$ and $y_i= 0$ otherwise. We further have a variable $x_{ij}$ that indicates whether $\mu_2(j) = i$. Note that  a lower bound on the switching cost is $\sum_{i\in C_1} (1-y_i) w_i$ since any center $i\in C_1$ that is closed in $\C_2$ incurs a switching cost of $w_i.$
Our linear program relaxation is the following. 
\begin{mini*}
{}{\sum_{i,j\in P_2} d(i,j) x_{ij}}{}{}
\addConstraint{\sum_{i\in P_2} x_{ij}}{= 1}{\forall j \in P_2}
\addConstraint{y_{i}}{\geq x_{ij}}{\forall i,j \in P_2}
\addConstraint{\sum_{i\in P_2} y_i}{\leq k}{}
\addConstraint{\sum_{i \in C_1} (1-y_i) w_i}{\leq S}{}
\addConstraint{x_{ij}, y_i}{\in [0, 1]}{\forall i, j \in P_2}
\end{mini*}

Here we allow for a total switching cost of $S$ and minimize the connection cost subject to that hard constraint.

We note that any rounding for $k$-median that opens a facility $i$ with probability $y_i$ (and assigns a client $j\in P_1$ to $\mu_1(j)$ if opened) satisfies the switching cost in expectation. Hence, prior results (e.g. the rounding in~\cite{CharikarL12}) imply a constant-factor approximation on the connection cost and a switching cost that is satisfied in expectation but not achieving the strong guarantees of \cref{thm:k-median}. To satisfy the switching cost, for any $i\in C_1$ that remains opened  $i\in C_2$, we make sure that  all the points from $P_1$ that were assigned to $i$ remain assigned to $i$; this is the reason for the additional $\cost(P_1, \C_1)$ in the bound on the cost of solution $\C_2$. 
We now give a high-level overview of the techniques of our rounding algorithm that extends the algorithm of~\cite{CharikarL12} to achieve these stronger guarantees on the switching cost.  
On a high-level the algorithm in~\cite{CharikarL12} works as follows. The variable $y_i$ is interpreted as the ``probability'' that point $i$ should be a center and the algorithm will open each point $i$ as a center with probability $y_i$ subject to the constraint that every point $j$ should have a center opened that is not too far (in expectation) compared to its opening cost $\sum_{i} x_{ij} d(i,j)$ in the linear program. The constraint that every client should have a ``nearby'' center can be abstractly modeled as a matroid constraint, and hence the rounding opens centers subject to a matroid constraint such that  each point $i$ is a center with probability $y_i$. We augment this high-level rounding strategy by formulating a second matroid constraint for the switching cost. This allows us to give a rounding that opens each center $i$ with probability $y_i$ subject to both matroid constraints, thus ensuring both that each point has a nearby center and the stronger guarantees on switching cost. In the detailed description of the algorithm, the fact that the constraints form two matroids is important for the proof of \cref{lem:integral} that relies on the integrality of the matroid intersection polytope. For a full description of the algorithm and proof please refer to Appendix~\ref{app:lprounding}.

%% file: app-k-cetner.tex
\section{Details of label-consistent $k$-center Algorithm}
\label{sec:app-k-center}
In this section, we provide an algorithm that, given two instances $P_1$ and $P_2$ of the $k$-center problem, where $P_1 \subseteq P_2$, a solution $\C_1 = (C_1,\mu_1)$ of the instance $P_1$, and a budget $S$ on the switching cost, outputs a solution $\C_2 = (C_2,\mu_2)$ of $P_2$ such that
\begin{itemize}
    \item $\swcost(\C_1,\C_2) \le S$, and
    \item $\cost_{cen}(P_2,\C_2) \le 6 \cdot \cost_{cen}(P_2,\C_2^*) $
\end{itemize}
where $\C^*_2=(C^*_2,\mu^*_2)$ denotes an (arbitrary) optimum solution for $P_2$ satisfying the given budget of at most $S$ on the switching cost.

\paragraph*{Description of the algorithm. }

Our algorithm works in two phases. In the first phase, we grow the cluster around each "old" center ($c\in C_1$). Then, we create new centers from the point set $P_2 \setminus P_1$ that are not covered by the grown clusters around old centers. In the next phase, we close some of the old centers by reassigning the corresponding cluster points to other old centers. Let us now elaborate on each of these two phases.

Before proceeding further, let us first introduce a few notations that we use to describe our algorithm. For each center $c \in C_1$, let us define its \emph{followers set} $\fol_1(c)$ to be the set of points in $P_1$ assigned to the center $c$ (by $\mu_1$). More specifically,
\[
\fol_1(c, r) := \left\{ p \in P_1 \mid \mu_1(p)=c,\; d(c,p)\leq r \right\}
\]
and $\fol_1(c):= \bigcup_r \fol_1(c, r)$.

Let $X$ be the underlying metric space. For any point $x \in X$ and any $r \ge 0$, let $\ball(x,r)$ denote the \emph{ball} of radius $r$ around $x$. Formally,
\[
\ball(x,r) := \left\{y \in X \mid d(x,y) \le r\right\}.
\]

\noindent \textbf{Phase 1: Growing phase and opening new centers. }For ease of presentation, we assume that we know the optimum value of the objective function. (Later in Remark~\ref{rem:assump}, we discuss how to remove this assumption.) Let $R := \cost(P_2, \C^*_2)$, where $\C^*_2=(C^*_2,\mu^*_2)$ denote an (arbitrary) optimum solution for $P_2$ that satisfies the budget of at most $S$ on the switching cost. 

Now, we grow the followers set for each center $c \in C_1$ and cover the remaining points. For that purpose, for each center $c \in C_1$, consider the set $P_2 \cap \ball(c,2R)$ and call these points of $P_2$ \emph{covered}. Let $U$ denote the set of \emph{uncovered} points of $P_2$. Formally,
\[
U:= P_2 \setminus \cup_{c \in C_1}\ball(c,2R).
\]

Next, we open a set of new centers from $U$ to cover all the points in $U$. In particular, we open a center at an arbitrarily chosen point $u \in U$, remove all the points in $\ball(u,2R)$ from $U$, and continue until $U$ becomes empty. 

Let $k'$ be the number of centers opened at the end of the above greedy procedure (i.e., $k' = |C_2|$). Thus, we are allowed to keep at most $k-k'$ centers from the set of old centers $C_1$. In the next phase, we keep $k-k'$ centers from $C_1$ and close the remaining.

\noindent \textbf{Phase 2: Closing and reassigning phase. }For each $c \in C_1$, let us define its \emph{weight} $w_c$ to be the size of its followers set $\fol_1(c,2R)$, i.e., $w_c : = |\fol_1(c,2R)|$. We say that a center $c \in C_1$ \emph{dominates} another center $c' \in C_1$ iff $c' \in \ball(c, 2R)$. Let $D(c)$ denote the set of all the centers dominated by $c$, i.e., 
\[
D(c) := \left\{c' \in C_1 \mid c' \in \ball(c, 2R) \right\}.
\]

Next, we keep (at most) $k-k'$ old centers open and close the remaining, and then reassign the points in the followers set of the closed centers. We start with all the centers in $C_1$ as unmarked. We then sort the centers $c \in C_1$ in the non-increasing order of their weights $w_c$. Next, we process them one by one in that sorted order. If a center is unmarked so far, we keep that center "temporarily" open (by adding them in a new set $T$) and mark (and close) all the centers that it dominates. Next, we process each temporarily opened center $c$. Consider the set $\ball(c,R)$ and let $c'$ be a center in $C_1\cap \ball(c,R)$ that has the highest weight among all the centers in $C_1 \cap \ball(c,R)$ (note that $c'$ could be different from $c$ because it could be dominated by a third center in $T$). Then, we add $c'$ to $C_2$ and remove $c$ from $T$. In the end, if the number of open centers is less than $k$, we open closed centers (from $C_1$) one by one in the above-mentioned sorted order until the number of open centers is $k$.

We describe the algorithm formally in Algorithm~\ref{alg:k-center} and provide an example in \Cref{fig:kcenter}. Let $C_2$ denote the set of centers opened by our algorithm. We define the assignment function $\mu_2$ for the center set $C_2$ as follows: For each point $p \in P_2$, if $\mu_1(p) \in C_2$ and $d(p,\mu_1(p)) \le 2R$, then set $\mu_2(p) = \mu_1(p)$; otherwise, set $\mu_2(p)$ to be the closest (breaking ties arbitrarily) center point in $C_2$.

For ease of presentation, in the algorithm's description, we assumed that the optimum objective value $R$ is known. This is a standard assumption and can be removed in a straightforward way. We explain the details and also running time analysis later.

\begin{algorithm}
	\begin{algorithmic}[1]
            \REQUIRE Point sets $P_1,P_2$ where $P_1 \subseteq P_2$, a $k$-center clustering $\C_1=(C_1,\mu_1)$ of $P_1$, and a non-negative value $R$.
		
		\ENSURE A $k$-center clustering $\C_2=(C_2,\mu_2)$ of $P_2$.
		
		
		\noindent \textbf{\underline{Phase 1:}}
            \STATE Construct the set $U = P_2 \setminus \cup_{c \in C_1}\ball(c,2R)$.
            \STATE Initialize $C_2\leftarrow \emptyset$.
            \WHILE{$U \ne \emptyset$}
                \STATE Choose a point $u \in U$ arbitrarily, and add $u$ to $C_2$. 
                \STATE $U \leftarrow U \setminus \ball(u, 2R)$.
		\ENDWHILE
  
            \noindent \textbf{\underline{Phase 2:}}
            \STATE Initialize all the centers in $C_1$ as unmarked.
            \STATE Sort the centers $c \in C_1$ in the non-increasing order of their weights $w_c = \left| \fol_1(c,2R) \right|$.
            \STATE Initialize $T \leftarrow \emptyset$.
            \WHILE{there exists a center in $C_1$ that is not marked}
                \STATE Let $c$ be the first point (in the sorted order) in $C_1$ that is unmarked.
                \STATE Add $c$ to $T$ and mark all the centers in $D(c)$.
            \ENDWHILE
            \FORALL{$c \in T$}
                \STATE Let $c' \in C_1 $ be the center with maximum weight (breaking ties arbitrarily) in $\ball(c,R)$, i.e., $c' = \arg \max_{i \in C_1 \cap \ball(c,R)}w_i$.
                \STATE Add $c'$ to $C_2$.
            \ENDFOR
            \WHILE{$|C_2| \ne k$}
                \STATE Let $c$ be the first point (in the sorted order with respect to the weight $w_c$) in $C_1$ that is not in $C_2$.
                \STATE Add $c$ to $C_2$.
            \ENDWHILE 

            \STATE Set $\mu_2$ as follows: For each point $p \in P_2$, if $\mu_1(p) \in C_2$ and $d(p,\mu_1(p)) \le 2R$, then set $\mu_2(p) = \mu_1(p)$; otherwise, set $\mu_2(p)$ to be the closest (breaking ties arbitrarily) center point in $C_2$.
		\STATE Return $\C_2 = (C_2 , \mu_2)$.

		\caption{Label consistent $k$-center clustering algorithm}
		\label{alg:k-center}
	\end{algorithmic}
\end{algorithm}

\paragraph*{Approximation guarantee. }Let us first consider the set of points in $P_2$ that are covered by the "old" centers (in $C_1$). We define
\[
O:=\cup_{c \in C_1}  \ball(c,R) \cap P_2 
\]
and let $N = P_2 \setminus O$. Let $\C^*_2=(C^*_2,\mu^*_2)$ denote an (arbitrary) optimum solution for $P_2$ such that $\swcost(\C_1,\C^*_2) \le S$. 

We define the followers sets of centers $c \in C_2$ ($c^* \in C^*_2$), denoted by $\fol_2(c)$ ($\fol^*_2(c^*)$), similar to $\fol_1(c)$ (the followers sets of centers $c \in C_1$). More specifically,
\begin{align*}
\text{for any }c \in C_2,\; \fol_2(c) &:= \left\{ p \in P_2 \mid \mu_2(p)=c \right\},\\
 \text{for any }c \in C^*_2,\;   \fol^*_2(c) &:= \left\{ p \in P_2 \mid \mu^*_2(p)=c \right\}.
\end{align*}

\begin{lemma}
    \label{lem:comp-new-center}
    $\left| C_2 \cap N \right| \le \left| C^*_2 \cap N \right|$.
\end{lemma}
\begin{proof}
    Consider any center $c \in C_2 \cap N$. Recall that in Phase 1, we have defined $U:=P_2 \setminus \cup_{c \in C_1} \ball(c,2R)$. Thus, by the construction of $C_2$, $c \in U$.
    
    We know that there exists a center $c^* \in C^*_2$ such that $c$ lies in the followers set of $c^*$, denoted by $\fol^*_2(c^*)$. Thus $d (c,c^*) \le R$. It then follows that $c^* \in N$; otherwise, there exists some $c'' \in C_1$ such that $d(c^*,c'') \le R$. Then, by the triangle inequality, 
    $
    d(c,c'') \le d(c,c^*) + d(c^*,c'') \le 2R
    $
    contradicting the fact that $c \in U$.

    Further, by the construction of $C_2$ (see Phase 1), for any two distinct $c,c' \in C_2 \cap N$, $d(c,c') > 2R$. Thus, for each $c \in C_2 \cap N$,  there exists (at least) a distinct $c^* \in C^*_2 \cap N$, implying $\left| C_2 \cap N \right| \le \left| C^*_2 \cap N \right|$.
\end{proof}

We now use the above lemma to argue that $\C_2 = (C_2,\mu_2)$ is a 6-approximate solution to the $k$-center problem on $P_2$. Recall, $R = \cost(P_2, \C^*_2)$.

\begin{lemma}
    \label{lem:k-cen-feasible}
    $\cost_{cen}(P_2,\C_2) \le 6 R.$
\end{lemma}
\begin{proof}
We first argue that in Phase 2 of our algorithm, all the centers $c \in C_1$ get marked after the addition of (at most) $ k - k'$ unmarked centers in $T$ (recall, $k'$ is the number of centers opened in Phase 1).

Recall that in Phase 1, we open $k'$ centers from the set $U = P_2 \setminus \cup_{c \in C_1} \ball(c,2R)$. Further, by the construction of $C_2$, for any center $c \in C_2 \cap N$, $c \in U$. Then, by Lemma~\ref{lem:comp-new-center},
$
\left| C^*_2 \cap O \right| \le k - \left| C_2 \cap N \right| = k - k'.
$
Further, observe that in the solution $\C^*_2$, any $c \in C_1$ must be covered by some center $c^* \in C^*_2 \cap O$. In other words, $c \in \fol^*_2(c^*)$, for some $c^* \in C^*_2 \cap O$. Thus, in Phase 2 of our algorithm, when we process a center $c \in C_1 \cap \fol^*_2(c^*)$ (where $c^* \in C^*_2 \cap O$) all the centers in the set $ C_1 \cap \fol^*_2(c^*)$ will be in the dependent set $D(c)$ and thus will be marked (since for any $c, c' \in C_1 \cap \fol^*_2(c^*)$, by the triangle inequality, $d(c,c') \le d(c,c^*) + d(c^*,c')\le 2R$). Hence, all the centers $c \in C_1$ get marked after at most $\left| C^*_2 \cap O \right| \le k - k'$ unmarked centers are added to $T$. Hence, $\C_2$ is a valid $k$-center solution for $P_2$.

Next, let us consider any center $c^* \in C^*_2 \cap \left( \cup_{c \in C_1} \ball(c,2R) \right)$. So there exists some $c \in C_1$ such that $d(c,c^*) \le 2R$. Then for all $p \in \fol^*_2(c^*)$, by the triangle inequality,
$
d(p,c) \le d(p,c^*) + d(c^*,c) \le R + 2R = 3R.
$

Now, either $c$ is added to the set $T$ (in Phase 2), in which case, there is a center $c_2 \in C_2$ (added during the execution of lines 14-16 of Algorithm~\ref{alg:k-center}) such that $d(c,c_2)\le R$. Then by the triangle inequality, $d(p,c_2) \le 3R + R = 4R$.

Otherwise (i.e., $c$ is not added to $T$), by Phase 2 of our algorithm, there must exist some $c' \in C_1$  such that
\begin{itemize}
    \item $c'$ is processed before $c$ in Phase 2 and is added to $T$, and
    \item $c \in D(c')$ (and thus gets marked while processing $c'$).
\end{itemize}
Then for all $p \in \fol^*_2(c^*)$, by the triangle inequality,
\begin{align*}
    d(p,c') &\le d(p,c) + d(c,c')
    \le 3R + 2R = 5R.
\end{align*}

Let $c_2$ be the center added to $C_2$ during the execution of lines 14-16 of Algorithm~\ref{alg:k-center} for $c' \in T$. Then, by the triangle inequality, $d(p,c_2) \le 5R + R = 6R$.

Next, let us consider any center $c^* \in C^*_2 \cap U$. By the construction of $C_2$ (in Phase 1 of our algorithm), there must exist a center $c \in C_2$ such that $d(c^*,c) \le 2R$. Then, for all $p \in \fol^*_2(c^*)$, again by the triangle inequality, $d(p,c) \le 3R$. This concludes the proof.
\end{proof}

It now only remains to show that the solution $\C_2$ satisfies the budget constraint of at most $S$ on the switching cost.

\begin{lemma}
    \label{lem:k-cen-switch}
    $\swcost(\C_1,\C_2) \le S$.
\end{lemma}

\begin{proof}
    Recall that 
    $
\swcost(\C_1,\C_2) = \left| \left\{ i \in P_1 \mid \mu_1(i) \neq \mu_2(i) \right\} \right|.
    $
It is straightforward to observe that given a set of centers $C_2$, to minimize the switching cost, for all the centers $c \in C_1 \cap C_2$, it must hold that for any $p \in \fol_1(c,2R)$, $\mu_1(p) = \mu_2(p)$. Thus, without loss of generality, we can say that $\swcost(\C_1,\C_2)$ is the number of points in $P_1$ that need to be (re-)assigned to a new center (in $C_2 \setminus C_1$) due to the closure of the corresponding centers in the solution $\C_2$, i.e.,
\begin{equation}
    \label{eq:switch}
    \swcost(\C_1,\C_2) = |P_1| - \sum_{c \in C_1 \cap C_2} w_c
\end{equation}

where $w_c = \left| \fol_1(c, 2 R) \right|$. Similarly,
\begin{align}
    \label{eq:switch-opt}
    \swcost(\C_1,\C^*_2) &= |P_1| - \sum_{c \in C_1 \cap C^*_2} |\fol_1(c)| \nonumber\\
    &\ge |P_1| - \sum_{c \in C_1 \cap C^*_2} w_c
\end{align}
where the last inequality follows since $\fol_1(c,2R) \subseteq \fol_1(c)$.

Next, we argue that $\sum_{c\in C_1 \cap C_2} w_c \ge \sum_{c \in C_1 \cap C^*_2} w_c$, which by~\cref{eq:switch} and~\cref{eq:switch-opt} implies $\swcost(\C_1,\C_2) \le \swcost(\C_1,\C^*_2)$.

For that purpose, let us consider the set $T$ after the execution of the while loop in lines 10-13. Let us consider any $c \in T$. We argue that there must exist $c^*_2 \in C^*_2$ such that
\begin{itemize}
    \item $c^*_2 \in \ball(c,R)$,
    \item $c^*_2 \in C^*_2 \cap O$,
    \item $c^*_2 \not \in \cup_{c \ne c' \in T} \ball(c',R)$.
\end{itemize}
The first item is straightforward to see. The second item is true by definition of $O$. To see the third item, for the sake of contradiction, let us assume $c^*_2 \in \ball(c', R)$ for some $c' \ne c$ in $T$. Without loss of generality, assume that $c'$ is processed after $c$ during the execution of lines 10-13 of Algorithm~\ref{alg:k-center}. By the triangle inequality, $d(c,c') \le d(c,c^*_2) + d(c^*_2,c') \le R+R = 2R$, which is not possible because then $c'$ will be marked while $c$ is being processed (since $c' \in D(c)$) and thus $c'$ cannot be added to $T$, leading to a contradiction. In the rest of the proof we will say that $c_2$ is related to $c$

As an immediate corollary of~\cref{lem:comp-new-center}, we get that $|C_2 \cap C_1| \ge |C^*_2 \cap O|$ and so $|C_2 \cap C_1| \ge |C^*_2 \cap C_1|$. So every $c^*_2 \in C^*_2 \cap C_1$ is related to some center in $C_2 \cap C_1$. Furthermore, consider a $c \in T$ and let $c^*_2 \in C^*_2$ be such that (i) $c^*_2 \in \ball(c,R)$, and (ii) $c^*_2 \not \in \cup_{c \ne c' \in T} \ball(c',R)$ and such that $c^*_2 \in C_1$. While processing $c \in T$ in lines 14-17 of Algorithm~\ref{alg:k-center}, suppose we add the center $c_2 \in C_2$ (note, $c_2 \in C_1 \cap C_2$). Then, by the construction, $w_{c_2} \ge w_{c^*_2}$. 

Furthermore, by lines 18-21 of Algorithm~\ref{alg:k-center}, the remaining $k-k'-|T|$ centers of $C_1 \cap C_2$ are added in the non-increasing order of their weights. Hence, we conclude that 
\[
\sum_{c \in C_1 \cap C_2} w_c \ge \sum_{c \in C_1 \cap C^*_2} w_c.
\]
It now directly follows from~\cref{eq:switch} and~\cref{eq:switch-opt} that 
\[
\swcost(\C_1,\C_2) \le \swcost(\C_1,\C^*_2) \le S
\]
which concludes the proof.
\end{proof}

Next, we remove the assumption of knowing the optimum objective value $R$ and analyze the running time.
\begin{remark}
    \label{rem:assump}
    For ease of presentation, in the algorithm's description, we assumed that the optimum objective value $R$ is known. This is a standard assumption and can be removed in a straightforward way. Since centers must be from the input set $P_2$, there are only ${|P_2| \choose 2}$ choices on the value of the optimum objective. Further, once we get a solution $\C_2$ for $P_2$, we can compute the value of $\swcost(\C_1,\C_2)$ in $O(n)$ time. Thus, we can perform a Binary search over all such possibilities as $R$ and finally return the solution that achieves the minimum objective value while satisfying the budget constraint of at most $S$ on the switching cost.
\end{remark}

\paragraph*{Time complexity. }Assuming computing distance between any two points $x,y$ in the metric space, i.e., the value of $d(x,y)$, takes $O(1)$ time, computing all the pairwise distances among the points in $P_2$ takes $O(n^2)$ time. Further, given the solution $\C_1$, computing the size of the followers set of all $c \in C_1$ takes $O(kn)$ time. For each choice of the value of $R$, constructing the set $U$ and opening centers in Phase 1 take $O(kn)$ time. In Phase 2, constructing the set $D(c)$ for all $c \in C_1$ takes $O(k^2)$ time. Further, in Phase 2, deciding which of the old centers (from the set $C_1$) to keep open (first adding centers to the set $T$ and then deciding which centers to add in $C_2$) takes $O(k^2)$ time. Hence, for each choice of $R$, Algorithm~\ref{alg:k-center} takes $O(kn)$ time. Thus, the overall time taken by performing a Binary search over all the possible choices of $R$ (as described in Remark~\ref{rem:assump}) is $O(n^2 + kn \log n)$. 

%% file: app-fastdp.tex
\section{Description of Fast Rounded DP and Analysis}\label{sec:dpfast}

Our final dynamic program has three dimensions. The first one indicates the name of the node on the tree, the second one is the number of centers we allow to be opened in the sub-tree of the node, and the last one is the maximum allowed connection cost of nodes connected within the sub-tree. As we explain later, we restrict the values of this connection cost dimension to a small set of rounded values in the end. The goal of the DP is to find the minimum switching cost solution that assigns all the points to a center in the sub-tree with the discussed requirements. If no such solution exists, we define the optimal switching cost as infinite.
\begin{align*}
 dp[\text{id}][k][D] = & \text{ Min switching cost in the subtree of node ``id'',} \\ & \text{ with at most $k$ open centers} \\ & \text{ and (rounded) connection cost at most $D$.}
\end{align*}

The initialization of the DP for the leaves is set to the weight of the point in the leaf if it is a new center 
and to zero otherwise (recall that $ k \geq 1$), independent of the connection cost $D$.  

For the internal nodes of the tree, one needs to distribute the number of open centers and (rounded) connection costs between the two children of the node.  
We describe how to compute $dp[\text{id}][k][D]$ for a node ``id'' with its two children named ``l'' and ``r'' and values $k$ and $D$. Let $n_l$ be the total weight of points in ``l'' and $n'_l$ be the total weight of points corresponding to centers in $C_1$ (old centers). We define $n_r$ and $n'_r$ analogously and let $d_{id}$ be the diameter of the points of ``id'', as computed by the tree embedding. The first option is to open zero centers in the ``l'' subtree and connect all points to one of the centers opened in the ``r'' subtree. The cost, in this case, is
\begin{align*}
    & dp[\text{r}][k][D - n_l \cdot d_{id}] + n'_l
\end{align*}
because the connection cost of moving the points in ``l'' is $n_l \cdot d_{id}$ leaving $D - n_l \cdot d_{id}$ for connection costs within ``r''. Furthermore, a switching cost of $n'_l$ is caused by moving the old points out of ``l'', in addition to the switching costs within ``r''. The analog cost of 
\begin{align*}
    & dp[\text{l}][k][D - n_r \cdot d_{id}] + n'_r
\end{align*}
holds for the case that we do not open any center on ``r''. The only remaining options are to open $1 \leq k' \leq k-1$ centers on the left node ``l'' with connection cost $0 \leq D' \leq D$ and $k-k'$ centers with the connection cost $D-D'$ on ``r''. The costs for these options are

\[
dp[\text{l}][k'][D'] + dp[\text{r}][k-k'][D-D'].
\]

The value of $dp[\text{id}][k][D]$ is the lowest among all the discussed options. 

This DP still computes the exact optimal switching cost. The time required is also still polynomial in $n$ because one of the dimensions optimized is the connection cost. Fortunately, we can afford a constant approximation in the connection cost given that the tree embedding already produces an $O(\log k)$-approximation in distances and, therefore, connection costs anyway. We, therefore, round connection cost values to powers of $(1+\eps)$ for some sufficiently small $\eps$, i.e., if at an internal node id, it is feasible to open $k'$ centers cause connection cost at most $D$ and switching cost $s = dp[\text{id}][k'][D]$ then we round $D$ up to the next power of $(1 + \eps)$, i.e., $\tilde{D} = (1+\eps)^{\lceil \log_{1+\eps} D \rceil}$ and note the feasibility of switching cost $s$, connection cost $\tilde{D}$ with opening $k'$ centers in the sub-tree below node ``id'' (in an implementation one of course simply stores the integer $\lceil \log_{1+\eps} D \rceil$ instead of $(1+\eps)^{\lceil \log_{1+\eps} D \rceil}$). Assuming a polynomial aspect ratio for distances and, therefore, also costs, this leaves only $O(\ \frac {\log n}{\eps})$ different cost values one needs to consider. It remains to analyze the impact of rounding to the power of $(1+\eps)$ at every step of the DP has on the approximation of the produced solution and how large or small of an $\eps$ one can choose, given that rounding errors can accumulate. 

In particular, if two values $D$ and $D'$ are rounded individually, then added up, and then rounded to the next power of $(1+\eps)$ again, the resulted value can differ from $D + D'$ by up to a $(1 + \eps)^2$ factor. This means that while rounding produces only a $(1+\eps)$ factor approximation in the leaves at level $\ell$ of the tree, this approximation error becomes as large as $(1 + \eps)^\ell$. However, by choosing $\eps = \frac{1}{101 \ell_{max}}$ where $\ell_{max} = O(\log k \log \Delta) = O(\log^2 n)$ is the depth of the tree produced by the tree embedding (assuming a polynomial aspect ratio $\Delta$), we can guarantee that the largest approximation is at most one percent because $(1 + \eps)^\ell \leq (1 + \frac{1}{101 \ell_{max}})^{\ell_{max}} \leq 1.01$.

Overall, the rounded DP is filling a three-dimensional array with $t=O(k)$ different possibilities for the first dimension, one for each node in the tree, $k+1$ different possibilities for how many centers (between $0$ and $k$) are opened, and $\frac{\log n}{\eps} = O(\polylog n)$ many possibilities for the rounded maximum connection cost. This results in a DP of $O(k^2 \polylog n)$ size. 
As described precisely above, to compute $dp[\text{id}][k'][D]$ there are up to $k' \leq k$ different possibilities to split up the $k'$ open centers between its two children and now $O(\polylog n)$ many possibilities for rounded connection costs in the left and right child resulting in a new rounded cost equal or below $D$ at the node $id$. Each of these options can be read off from the sum of two previous DP values, and the best of these $O(k' \polylog n)$ options is therefore computed in $O(k' \polylog n)$ time at each of the $O(k^2 \polylog n)$ nodes. The total computation time is, therefore, $O(k^3 \polylog n)$, as claimed. This completes the description and analysis of the accelerated DP mentioned in~\cref{thm:RoundedTreeDP}.

Next, we provide the pseudocode of our algorithm referred to in~\cref{thm:dpmain}. Although in the pseudocode below (\cref{alg:k-median}), we return the cost of a $k$-median clustering $\C_2$ of $P_2$, it can be amended straightforwardly to return the corresponding clustering.

\begin{algorithm}
	\begin{algorithmic}[1]
            \REQUIRE Point sets $P_1,P_2$ where $P_1 \subseteq P_2$, a $k$-median clustering $\C_1=(C_1,\mu_1)$ of $P_1$, and a non-negative number $S$.
		
		\ENSURE The cost of a $k$-median clustering $\C_2$ of $P_2$ such that $\swcost(\C_1,\C_2) \le S$.
		
            \STATE For each $i \in C_1$, assign weight $w_i = |\{j \in P_1 \mid \mu_1(j)=i\}|$.

            \STATE Perform $k$-median++ algorithm on the set $P_2 \setminus P_1$ and get a solution $\tilde{\C}_2=(\tilde{C}_2 , \tilde{\mu}_2)$, where $|\tilde{C}_2| = O(k)$.

            \STATE For each $i \in \tilde{C}_2$, assign weight $w_i = |\{j \in P_2 \setminus P_1 \mid \tilde{\mu}_2(j)=i\}|$.

            \STATE Consider $C= C_1 \cup \tilde{C}_2$ with each $i \in C$ having weight $w_i$. 
            
            \STATE Embed the points in $C$ into a tree of depth $\ell_{max}=O(\log k \log \Delta)$ using~\cite{fakcharoenphol2003tight}. Let $\mathcal{T}$ be that tree with root node 'root'.
            
            \STATE Set $\eps = 1/101 \ell_{max}$.
            
            \STATE $\mathcal{R} := \left\{0, 1, (1+\eps), (1+\eps)^2,\cdots,(1+\eps)^{\lceil \log_{1+\eps} \Delta \rceil}\right\}$.

            \STATE Initialize $dp[id][j][D] \leftarrow \infty$ for all node ``id'' in $\mathcal{T}$, for all $j \in \{1,2, \cdots, k\}$, for all $D \in \mathcal{R}$.

            \STATE \textbf{Base Case:} Set $dp[id][j][D] \leftarrow 0$ for all leaves ``id'' in $\mathcal{T}$, for all $j \in \{1,2, \cdots, k\}$, for all $D \in \mathcal{R}$.
            
            \FOR{$\ell=\ell_{max}-1, \cdots, 0$}
            		\FORALL{node ``id'' at level $\ell$ of tree $\mathcal{T}$}
            			\FOR{$j=1,\cdots,k$}
            				\FORALL{$D \in \mathcal{R}$}
            					\STATE Let ``l'' and ``r'' be the left and right children of the node 'id.'
            					
            					\STATE Let $n_l$ and $n_r$ be the total weights of the points (leaves) $\in C$ in the subtree of ``l'' and ``r'', respectively.
            					
            					\STATE Let $n'_l$ and $n'_r$ be the total weights of the points (leaves)$\in C_1$ in the subtree of ``l'' and ``r'', respectively.
            					
            					\STATE Let $d_{id}$ be the diameter of the points in the subtree of the node ``id''.

                                    \STATE Let $D_l$ and $D_r$ be the smallest integers in $\mathcal R$ that are at least $D - n_l\cdot d_{id}$ and $D - n_r\cdot d_{id}$, respectively.

            					\STATE $dp[id][j][D] = \min\left\{\left(dp[l][j][D_r] + n'_r\right), \left(dp[r][j][D_l] + n'_l\right)\right\}$.
            					
            					\FOR{$k' = 1,\cdots, j-1$}
            						\FORALL{$D' \in \mathcal{R}$ such that $ D' \le D$}

                                            \STATE Let $D''$ be the smallest integer in $\mathcal R$ that is at least $D - D'$.
                                            
            							\STATE $dp[id][j][D] = \min\left\{dp[id][j][D], \left(dp[l][k'][D'] + dp[r][j - k'][D'']\right)\right\}$.
            						\ENDFOR
            					\ENDFOR

            				\ENDFOR
            			\ENDFOR
            		\ENDFOR
            \ENDFOR
            
		\STATE Return Minimum $D$ such that $dp[root][k][D] \le S$.

		\caption{$O(\log k)$ approximation algorithm for the label consistent $k$-median clustering}
		\label{alg:k-median}
	\end{algorithmic}
\end{algorithm}

\section{Proof of \cref{thm:dpmain}} \label{app:dpmainproof}
For the sake of convenience, let us start by restating the theorem:

\thmdpmain*

\begin{proof}
Let $\alpha$ represent the constant approximation guarantee of the algorithm used to find $O(k)$ centers from the points in $P_2\setminus P_1$ (denoted by $\mathcal{A}$ in the main part). Moreover let $\beta, C_f$ denote the approximation factor and the final solution of our Dynamic Programming algorithm, which is $O(\log k)$-approximate due to \cref{thm:RoundedTreeDP}. Also let $P_3$ denote the weighted points obtained by moving the points in $P_1$ and $P_2 \setminus P_1$ to $C_1$ and the centers computed on $P_2 \setminus P_1$ , respectively. Recall that we run our DP on $P_3$. We use $\OPT(X)$ for a set of points $X$ to denote the cost of the optimum solution on the points in $X$. We observe that
\begin{align} \label{ine:dpapprox1}
    \cost(P_2, C_f) \leq \alpha \OPT(P_2) + \cost(P_1, C_1) + \cost(P_3, C_f) \,,
\end{align}
where the first term is an upper bound on the cost of moving the points in $P_2 \setminus P_1$, the second term is an upper bound on the cost of moving the points in $P_1$, and the third term is the cost of the solution of DP on $P_3$.
Similarly we get that
\begin{align} \label{ine:dpapprox2}
    \OPT(P_3) \leq (\alpha\OPT(P_2) + \cost(P_1, C_1)) + \OPT(P_2)\,,
\end{align}
where the first term is due to the cost of moving points from $P_2$ and $P_1$ (as above), and the second term is just the optimal cost of the original points, $P_2$. From the definition of $\beta$ we have that $\cost(P_3, C_f) \leq \beta \OPT(P_3)$. Combining with \cref{ine:dpapprox2} we get that
\begin{align*}
    \E[\cost(P_3, C_f)] \leq  \beta\left[(1+\alpha) \OPT(P_2) + \cost(P_1, C_1) \right]
\end{align*}
Combining with \cref{ine:dpapprox1} we get that 
\begin{align*}
    \E[\cost(P_2, C_f)] \leq \alpha \OPT(P_2) + \cost(P_1, C_1) + \beta\left[(1+\alpha) \OPT(P_2) + \cost(P_1, C_1) \right]\,.
\end{align*}
Since $a \in O(1)$ and $b \in O(\log k)$, we get that
\begin{align*}
    \E[\cost(P_2, C_f)] \leq O(\log k) \left(\OPT(P_2) + \cost(P_1, C_1) \right)
\end{align*}
    
We continue by analyzing the running time and the success probability. The point reduction from $P_2$ to $P_3$ can be done in $O(nk)$ time, due to the running time of $k$-median++. Computing a tree embedding on the resulting $O(k)$ points requires $O(k^2 \polylog n)$ time and computing the DP takes $O(k^3 \polylog n)$ time. To achieve the claimed result with high probability, $10 \log_2 n$ independent runs are performed, their cost is computed, and the best one is output. Given that each run has a probability of at most $1/2$ to output a solution twice as large as the expectation, the probability this happens in each run resulting in a bad output is polynomially small, i.e., at most $2^{-10 \log_2 n} = 1/n^{10}$. The bound on the switching cost guaranteed by \cref{thm:RoundedTreeDP} remains valid as we do not modify the assignments made by DP.
\end{proof}

%% file: app-kmedian.tex
\section{Rounding $k$-Median Linear Programming Relaxation}\label{app:lprounding}
Let $(x^*,y^*)$ be an optimal solution to the linear program presented in~\cref{sec:lprelaxation}. We give efficient rounding procedures of the fractional solution $(x^*, y^*)$ that yield \cref{thm:k-median}. 

We build on the rounding proposed in~\cite{CharikarL12} and extend it to deal with the additional constraint on the switching cost. We use the following notation.   For a client $j\in P_2$, we let $\dav(j) = \sum_{i\in P_2} x^*_{ij} d(i,j)$ be the average connection cost of this client in the LP solution. For $j\in P_2, r\in \mathbb{R}$, we also define $\ball(j, r) = \{j'\in P_2 \mid d(j', j) \leq r\}$ to be the set of clients in the ball of radii $r$ centered at $j$. For a general subset $U \subseteq P_2$, we let $\vol(U) = \sum_{i\in U} y^*_i$ be the \emph{volume} of that set. 

\paragraph*{Filtering and a matching of bundles.} Using fairly standard techniques for these problems, such as filtering and bundling, the authors in~\cite{CharikarL12} define the following sets:
\begin{itemize}
    \item A subset of $P'_2 \subseteq P_2$ of the clients satisfying
    \begin{itemize}
        \item[(i)] For any $j, j'\in P'_2$ with $j\neq j'$ we have \\$d(j,j')\geq 4 \max\{\dav(j), \dav(j')\}$.
        \item[(ii)] For any $j' \in P_2 \setminus P_2'$, there is a client $j\in P_2'$ such that $\dav(j) \leq \dav(j')$ and $d(j,j') \leq 4 \dav(j')$.
    \end{itemize}
\end{itemize}
In words, the clients in $P_2'$ are ``far away'' from each other, and at the same time, any client not in $P_2'$ is close to some client in $P_2'$. This is Lemma~1 in~\cite{CharikarL12}. 
Moreover, with each $j\in P_2'$, define $R_j = \frac{1}{2} \min_{j'\in P_2', j' \neq j} d(j,j')$ to be half the distance of $j$ to its nearest neighbor in $P_2'$. 
The authors define, for each $j\in P_2'$, a bundle $U_j \subseteq \ball(j,  R_j)$ with the following guarantees: 
\begin{itemize}
    \item[(i)] $1/2 \leq (1- d_{av}(j)/R_j) \leq \vol(U_j) \leq 1$ for all $j\in P'$ and
    \item[(ii)] $U_j \cap U_{j'} = \emptyset$ for all $j,j' \in P', j \neq j'$.
\end{itemize}
This is Lemma~2 in that paper. 

Finally, they construct a matching $M$ over the clients in $P_2'$ greedily as follows: while there are at least two unmatched clients in $P_2'$, pick the two closest unmatched clients in $P_2'$ and match them.

\paragraph*{New rounding procedure.}
From~\cite{CharikarL12}, we thus have the clients $P_2' \subseteq P_2$, bundles $(U_j)_{j\in P_2'}$, distances $R_j$, and the matching $M$. The basic rounding of~\cite{CharikarL12} is obtained by defining a distribution over centers $C_2$ that satisfies the following properties:
\begin{enumerate}
    \item  We always have $|C_2| \leq k$.
    \item For a point $i\in P_2$, we have $\Pr[i \in C_2] = y_i$.
    \item For a bundle $U_j$ with $j\in P'_2$, we always have $|U_j \cap C_2| \leq 1$.
    \item For two bundles $U_j, U_{j'}$ that are matched by $M$, we always have $\left|(U_j \cup U_{j'}) \cap C_2\right| \geq 1. $
\end{enumerate}
These properties are sufficient to achieve a constant-factor approximation on the connection cost (and satisfying the switching cost in expectation). 
Our rounding differs in that we give a distribution over solutions $C_2$ that satisfy all the above properties and, in addition, it satisfies the following property:
\begin{enumerate}
    \item [(5)] Order the facilities in $C_1 = \{1, \ldots, k\}$ so that $w_1 \geq w_2 \geq  \ldots  \geq w_k$, and let $C_1(\ell) = \{1, 2, \ldots, \ell\}$ be the prefix of the first $\ell$ centers. Then, we always have
        \[
            |C_1(\ell) \setminus C_2| \leq  \left\lceil \sum_{i \in \{1, \ldots, \ell\}} (1-y^*_i)\right\rceil 
        \]
        for every prefix $\ell = 1, \ldots, k$.
\end{enumerate} 

Given this distribution, our basic rounding algorithm proceeds by sampling the solution $C_2$ and it defines the assignment $\mu_2$ as follows: for every $j\in P_1$ so that $\mu_1(j) \in C_2$ we let $\mu_2(j) = \mu_1(j)$ and for remaining clients $j$ we let $\mu_2(j)$ equal its closest center in $C_2$. The definition of $\mu_2$ ensures that the switching cost equals $\sum_{i\in C_1 \setminus C_2} w_i$. The bound on the connection cost now follows from Properties $1-4$ by the same arguments as in \cite{CharikarL12} and is argued in the end of the analysis. The novel part is to first argue that we can efficiently sample from a distribution that satisfies properties $1-5$ and that yields the desired bounds on the switching cost. Following the analysis of our basic rounding, we explain, in the end of this section, how the basic rounding scheme is slightly modified to prove the two statements of~\cref{thm:k-median}.

\paragraph*{Efficiently sampling for a distribution satisfying properties $1-5$.}

We formulate a polynomial-sized polytope that expresses the five properties. We then show that this polytope is integral and since it contains $y^*$ we can thus in polynomial time write $y^*$ as a convex solution to integral solutions, and sampling an integral solution with probability that equals its coefficient then gives the desired distribution.

The polytope is  defined by the following constraints:
\begin{itemize}
    \item \(\sum_i y_i \leq k \);
    \item \(\sum_{i\in U_j} y_i  \leq 1\) { for all $j\in P_2'$}; 
    \item \(\sum_{i \in U_j \cup U_{j'}} y_i  \geq 1 \) { for all $j, j' \in P_2'$ that are matched in $M$};
    \item \(\sum_{i \in C_1(\ell)} (1-y_i) \leq \left\lceil \sum_{i \in \{1, \ldots, \ell\}} (1-y^*_i)\right\rceil\) {for every prefix $\ell = 1,\ldots, k$ of $C_1$}; and
    \item \( 0 \leq y_i \leq 1\) for every $i\in P_2$.
\end{itemize}

Now, we are ready to prove that the above polytope is integral.

\begin{lemma}\label{lem:integral}
    Any extreme point to the above polytope is integral. 
\end{lemma}
\begin{proof}
    We will use the integrality of the so-called matroid intersection polytope in the special case of laminar matroids. Recall that a family $\cL$ of sets is \emph{laminar} if for any two sets $A,B \in \cL$ we either have that one is a subset of the other (i.e., $A \subseteq B$ or $B \subseteq A$) or they are disjoint $(A \neq B)$. The characterization of the integral polytope of the intersection of two laminar matroids now implies the following (see e.g., Theorem 41.11 and Theorem 5.20 in~\cite{Schrijver03}): \begin{quote}
        Consider an extreme point $x$ defined by tight constraints
        \begin{align*}
            \sum_{i\in S} x_i &= a_S \qquad S \in \cL_1  \\
            \sum_{i\in S} x_i & = b_S \qquad S \in \cL_2 \\
            0 \leq x_i & \leq 1 \,.
        \end{align*}
        Then $x$ is integral if $\cL_1$ and $\cL_2$ are laminar sets and the right-hand-sides $(a_S)_{S\in \cL_1}, (b_S)_{S\in \cL_2}$ are integral.
    \end{quote}
    Now to use the above fact in our setting, consider an extreme point $y'$ and consider the subset of constraints that are tight.
    First, consider the tight constraints of one of the following types:
    \[
        \sum_{i\in U_j} y_i  = 1, \quad \sum_{i \in U_j \cup U_{j'}} y_i = 1
    \]
    and the constraint $\sum_i y_i = k$ if tight. These form a laminar set $\cL_1$, and every right-hand-side is integral. 
    
    Second, consider the tight constraints of the type
    \[
        \sum_{i \in \{1, \ldots, \ell\}} (1-y_i)  =  \left\lceil \sum_{i \in \{1, \ldots, \ell\}} (1-y^*_i)\right\rceil 
    \]
    which equivalently can be written as
    \[
         \sum_{i\in \{1, \ldots \ell\}} y_i = \ell -  \left\lceil \sum_{i \in C_1(\ell)} (1-y^*_i)\right\rceil\,.
    \]
    These tight constraints again form a laminar set $\cL_2$ with integral right-hand-sides. Combining this with the above-mentioned fact thus implies that $y'$ must be integral. 
\end{proof}
\paragraph*{Bounding the switching cost.} We now show how the fifth property (5) bounds the switching cost.

\begin{lemma}
    The switching cost of any extreme point is at most $S + \max\{w_i \mid i\in C_1,  0 < y^*_i< 1\}$. Furthermore, we can achieve a switching cost of at most $S$ by opening up at most one additional center.
    \label{lem:switching_cost_kmedian}
\end{lemma}
\begin{proof}
    For an extreme point $y'$ we have
    \[
    \sum_{i \in C_1(\ell)\}} (1-y'_i) \leq \left\lceil \sum_{i\in C_1(\ell)} (1-y^*_i)\right\rceil 
    \]
    for every $\ell =1 ,\ldots, k$.
    Let $i_0$ be the first index that satisfies $y'_{i_0} = 0$ and $0 < y^*_{i_0}$. If there is no such index $i_0$ then the switching cost of $y'$ is at most that of $y^*$ and thus at most $S$. 
    Hence, we can focus on the more interesting case when such an index $i_0$ exists and further that we have $y^*_{i_0} < 1$ since otherwise $y'_{i_0}$ must equal $1$. As $y'_i = 1$ whenever $y^*_i > 0$ for $i< i_0$, the selection of $i_0$ and the above inequality implies
    $
     \sum_{i \in C_1(\ell)} (1-y''_i) \leq  \sum_{i \in C_1(\ell)} (1-y^*_i) 
    $
    for every prefix $\ell = 1,\ldots, k$m 
    where $y''$ is the same vector $y'$  except that we set $y''_{i_0} = 1$
    Now the switching cost $\sum_{i\in C_1} (1-y''_i) w_i $ can be written as the following telescoping sum:
    $
        \sum_{\ell  = 1}^{k} \sum_{i \in C_1(\ell)\}} (1-y''_i) (w_{\ell}  - w_{\ell+1})\,,
    $
    where for notational convenience we let $w_{k+1}  = 0$. Using $\sum_{i \in C_1(\ell)} (1-y''_i) \leq  \sum_{i \in C_1(\ell)} (1-y^*_i)$, we thus have
   \begin{align*} 
        \sum_{\ell  = 1}^{k} \sum_{i \in C_1(\ell)} (1-y''_i) & (w_{\ell}  - w_{\ell+1}) 
         \leq         \sum_{\ell  = 1}^{k} \sum_{i \in C_1(\ell)} (1-y^*_i) (w_{\ell}  - w_{\ell+1}) \\
        & = \sum_{i\in C_1} (1-y^*_i) w_i \leq S\,.
    \end{align*}
    In other words, the switching cost of $y''$ is at most $S$. As the difference between the switching costs of $y'$ and $y''$ is only $i_0$, the lemma follows. 
\end{proof}

\paragraph*{Bounding the connection cost.}
To bound the connection cost, we follow the same arguments as in the proof of  Lemma~4 in~\cite{CharikarL12}. The following shows that the expected connection cost of our mapping $\mu_2$ that either sends a client to the same center as $\mu_1$ or to its closest center is at most $\cost(P_1, \C_1) + 10\cdot \sum_{j\in P_2} \dav(j)$. As the LP is a relaxation, we have $\sum_{j\in P_2} \dav(j) \leq \OPT$ and so we have the bound $\cost(P_2, \C_2)\leq 10\cdot \OPT + \cost(P_1, \C_1)$  on the connection cost.  The proof of the lemma appears in \cref{sec:app-lp-cost}.
\begin{lemma}
  \label{lem:lp-cost}
   The expected distance of a client $j\in P_2$ to its closest center is at most $10\cdot \dav(j)$.
\end{lemma}

\paragraph*{Proof of~\cref{thm:k-median}.}
\cref{lem:switching_cost_kmedian} imply that we can exactly satisfy the switching cost with $k+1$ opened centers. Opening up one more center can only decrease the switching cost and so we maintain the bound $\cost(P_2, \C_2)\leq 10\cdot \OPT + \cost(P_1, \C_1)$. This  implies the first part of \cref{thm:k-median}. For the second part, we proceed as follows. We let $H = \{ i\in C_1 \mid w_i \geq \varepsilon S\}$ be all ``heavy'' centers in the first solution. Note that any solution that satisfies the switching cost can close at most $\lfloor 1/\varepsilon \rfloor$ many such centers. We can thus in time $O(|C_1|^{\lfloor 1/\varepsilon \rfloor}) = O(|P_2|^{1/\varepsilon})$ guess the heavy centers that are closed by the optimal solution. For each guess, we solve the LP relaxation where we set $y_i = 1$ for all $i\in H$ that were not closed and $y_i = 0$ for all closed $i\in H$. We then run our rounding algorithm that produces a solution of switching cost of at most $(1+\varepsilon)S$ by \cref{lem:switching_cost_kmedian}. Among all the solutions created, we return the one with the smallest connection cost. For the guess of closed centers in $H$ that agree with the optimal solution, our analysis of the connection cost guarantees $\cost(P_2, \C_2)\leq 10\cdot \OPT + \cost(P_1, \C_1)$. In particular, this holds for the returned solution of the lowest cost. We have thus in polynomial time (for every fixed $\varepsilon$) constructed a solution of switching cost at most $(1+\varepsilon)S$ and connection cost $\cost(P_2, \C_2)\leq 10\cdot \OPT + \cost(P_1, \C_1)$.  This yields the second part of \cref{thm:k-median}.

%% file: app-integrality.tex
\section{Integrality gap}\label{sec:lpintegralitygap}
In this section, we establish an unbounded integrality gap of the standard LP formulation of the label consistent $k$-median problem when one opens at most $k$ centers and must satisfy the switching cost budget, thus establishing~\cref{lem:integrality}. Consider the following instance: 

\begin{itemize}
    \item $C_1$ consists of $k$ centers with pairwise distance $1$ and $P_1$ consists of $Mk$, each center in $C_1$ has $M$ co-located points.
    \item $P_2$ contains, in addition to the points in $P_1$, $2$ points that are infinitely far from the centers in $C_1$ and at distance $D$ from each other.
    \item The allowed switching cost is at most $S = 2M-1$.
\end{itemize}

By construction, any solution that satisfies the switching cost must open $k-1$ centers in $C_1$, and thus, only one center will serve the two points in $P_2 \setminus P_1$. This leads to a connection cost of at least $D$.

On the other hand, a fractional solution that sets $y_i =  (k-2)/k + 1/(Mk) $ for $i\in C_1$ will satisfy the fractional switching cost of the LP since
\[
 \sum_{i\in C_1} (1-y_i) M = kM - kM \cdot ((k-2)/k + 1/MK) = 2M - 1\,.
\]
We can thus fractionally open up two centers, each with fractional mass  $(2-1/M)/2$, that is co-located with the two points in $P_2 \setminus P_1$. This leads to a fractional connection cost of $D/M + k M$, where we naively upper bound the connection cost of the points in $P_1$ by $1$. Letting $D$ and $M$ tend to $\infty$ now yields the unbounded integrality gap.

%% file: app-lp-cost.tex
\section{Proof of Lemma~\ref{lem:lp-cost}}
\label{sec:app-lp-cost}

\begin{proof}
    We first consider the case when $j \in P_2'$. Let $j_2$ be the client in $P_2'\setminus \{j\}$  that is closest to $j$. 
    
    Suppose that $j$ is matched to $j_2$. We have that $R_j = d(j,j_2)/2\}$.
    We have that with probability $\vol(U_j)$, a random center in the ball $\ball(j, R_j)$ is opened which has expected connection cost at most $\dav(j)$. With the remaining probability, we are guaranteed that a center in $U_{j_2}$ is opened, which is at distance at most $d(j,j_2) + R_{j_2}  \leq 3 R_j$ as $R_{j_2} \leq R_j$. As $\vol(U_j) \geq 1- \dav(j)/R_j$ we thus get that the expected distance to the closest center in this case it at most $4 \dav(j)$. 
    
    Now suppose that $j$ is not matched to $j_2$. Then $j_2$ must have been matched to a third center $j_3$ before $j$ was matched. By greedy construction of the matching, we thus have $2 \max\{R_{j_2}, R_{j_3}\} \leq d(j_2, j_3) \leq d(j_1, j_2) = 2R_j$.
    Similarly to the previous case, we have that the connection cost of $j$ is at most $\dav(j)$ with probability $\vol(U_j)$. Otherwise, we use the fact that at least one center is opened in $U_{j_2} \cup U_{j_3}$. Any such center is within distance $d(j, j_2) + d(j_2, j_3) + \max\{R_{j_2}, R_{j_3}\} \leq 5R_j $ of $j$. As $\vol(U_j) \geq 1-\dav(j)/R_j$, we have that the expected distance to $j$'s closest center in this case is at most $6 \dav(j)$.
    
    It follows that any $j\in P_2'$ has expected distance to its closest center is  at most $6 \dav(j)$. Now for a client $j' \in P_2 \setminus P_2'$, there is a client $j \in P_2'$ so that $d(j',j)\leq 4 \dav(j')$ and $\dav(j) \leq \dav(j')$. Hence, by the triangle inequality, the expected distance to the closest center for $j'$ is at most $4 \dav(j') + 6 \dav(j) \leq 10 \dav(j')$, which completes the proof of the lemma.
\end{proof}